\documentclass[11pt]{llncs}

\usepackage{amssymb,amsmath,stmaryrd}
\usepackage{gastex}
\usepackage{paralist,fullpage}
\usepackage{comment}
\newcommand{\nat}{\mathbb{N}}
\newcommand{\integ}{\mathbb{Z}}

\newcommand{\mathN}{\mathbb{N}}
\newcommand{\integers}{\mathbb{Z}}

\newcommand{\MeanPayoff}{{\sf MeanPayoff}}
\newcommand{\EL}{{\sf EL}}
\newcommand{\MP}{{\sf MP}}
\newcommand{\Plays}{{\sf Plays}}
\newcommand{\Prefs}{{\sf Prefs}}

\newcommand{\Last}{{\sf Last}}
\newcommand{\outcome}{\mathsf{outcome}}

\def\abs#1{\ensuremath{\lvert #1\rvert}}
\newcommand{\tuple}[1]{\langle #1 \rangle}

\newcommand{\PosEnergy}{{\sf PosEnergy}}

\renewcommand{\qed}{\hfill \ensuremath{\Box}}

\newcommand{\Flow}{{\sf Flow}}

\newcommand{\Circuit}{{\sf Circuit}}

\title{Generalized Mean-payoff and Energy Games}
\author{Krishnendu Chatterjee\inst{1} \and Laurent Doyen\inst{2} \and Thomas A. Henzinger\inst{1} \and Jean-Fran\c{c}ois Raskin\inst{3}}

\institute{
IST Austria (Institute of Science and Technology Austria) \\
\and LSV, ENS Cachan \& CNRS, France \\
\and D\'epartment d'Informatique, Universit\'e Libre de Bruxelles (U.L.B.) 
}

\begin{document}

\maketitle

\begin{abstract}
In mean-payoff games, the objective of the protagonist is to ensure that the 
limit average of an infinite sequence of numeric weights is nonnegative. 
In energy games, the objective is to ensure that the running sum of weights 
is always nonnegative. Generalized mean-payoff and energy games replace 
individual weights by tuples, and the limit average (resp.\ running sum) 
of each coordinate must be (resp.\ remain) nonnegative. These games have 
applications in the synthesis of resource-bounded processes with multiple 
resources.

We prove the finite-memory determinacy of generalized energy games and show 
the inter-reducibility of generalized mean-payoff and energy games for 
finite-memory strategies. We also improve the computational complexity 
for solving both classes of games with finite-memory strategies: while 
the previously best known upper bound was EXPSPACE, and no lower 
bound was known, we give an optimal coNP-complete bound. 
For memoryless strategies, we show that the problem of deciding 
the existence of a winning strategy for the protagonist is NP-complete. 
\end{abstract}

\section{Introduction}

\noindent{\em Graph games and multi-objectives.} 
Two-player games on graphs are central in many applications
of computer science.  For example, in the synthesis problem, implementations are 
obtained from winning strategies in games with a qualitative objective such as 
$\omega$-regular specifications~\cite{RW87,PnueliR89,AbadiLW89}.
In these applications, the games have a qualitative (boolean) objective 
that determines which player wins. 
On the other hand, games with quantitative objective which are natural models 
in economics (where players have to optimize a real-valued payoff) 
have also been studied in the context of automated design~\cite{Sha53,Condon92,ZP96}.
In the recent past, there has been considerable interest in the design of 
reactive systems that work in resource-constrained environments 
(such as embedded systems). 
The specifications for such reactive systems are quantitative, and these 
give rise to quantitative games. 
In most system design problems, there is no unique objective to be optimized, 
but multiple, potentially conflicting objectives.
For example, in designing a computer system, one is interested not
only in minimizing the average response time but also the average power 
consumption.
In this work we study such multi-objective generalizations of the 
two most widely used quantitative objectives in games, namely,
\emph{mean-payoff} and 
\emph{energy} objectives~\cite{EM79,ZP96,CAHS03,BFLMS08}.

\smallskip\noindent{\em Generalized mean-payoff games.} 
A {\em generalized mean-payoff game} is played on a finite weighted game graph by two players. 
The vertices of the game graph are partitioned into positions that belong to 
Player~$1$ and positions that belong to Player~$2$. 
Edges of the graphs are labeled with $k$-dimensional vectors $w$ of integer values, i.e., 
$w \in \mathbb{Z}^k$. The game is played as follows. A pebble is placed on a 
designated initial vertex of the game graph. The game is played in rounds in which
the player owning the position where the pebble lies
moves the pebble to an adjacent position of the graph using an outgoing edge. 
The game is played for an infinite number of rounds, resulting in an infinite 
path through the graph, called a play. The value associated to a play is the 
mean value in each dimension of the vectors of weights labeling the edges of the play.
Accordingly, the winning condition for Player~1
is defined by a vector of integer values $v \in \mathbb{Z}^k$ that specifies a 
threshold for each dimension. A play is winning for Player~$1$ if its vector 
of mean values is at least~$v$. All other plays are winning for Player~$2$,
thus the game is zero-sum. We are interested in the problem of deciding
the existence of a finite-memory winning strategy for Player~$1$ in generalized 
mean-payoff games. Note that in general infinite memory may be required to win
generalized mean-payoff games, but for practical applications such as the 
synthesis of reactive systems with multiple resource constraints, the 
generalized mean-payoff games with finite memory is the relevant model.
Moreover, they provide the framework for the  
synthesis of specifications defined by~\cite{AlurDMW09,Concur10}, and 
the synthesis question for such specifications under \emph{regular (ultimately periodic)} 
words correspond to generalized mean-payoff games with finite-memory strategies.

\smallskip\noindent{\em Generalized energy games.} 
In generalized energy games, the winning condition for Player~1 requires that,
given an initial credit $v_0 \in \mathbb{N}^k$, the sum of $v_0$ and all the 
vectors labeling edges up to position $i$ in the play is nonnegative, for all 
$i \in \mathbb{N}$. The decision problem for generalized energy games
asks whether there exists an initial credit $v_0$ and a strategy for Player~1 
to maintain the energy nonnegative in all dimensions against all strategies 
of Player~2.

\smallskip\noindent{\em Contributions.}
In this paper, we study the strategy complexity and computational complexity 
of solving generalized mean-payoff and energy games. 
Our contributions are as follows.

\noindent
First, we show that generalized energy and mean-payoff games 
are determined when played with finite-memory strategies, however,  they are not 
determined for memoryless strategies. 
For generalized energy games determinacy under finite-memory coincides 
with determinacy under arbitrary strategies (each player has a
winning strategy iff he has a finite-memory winning strategy).
In contrast, we show for generalized mean-payoff games that determinacy under finite-memory 
and determinacy under arbitrary strategies do not coincide. 
Thus with finite-memory strategies these games are determined, they correspond 
to the synthesis question with ultimately periodic words, and enjoy pleasant 
mathematical properties like existence of the limit of the 
mean value of the weights, and hence we focus on the study of 
generalized mean-payoff and energy games with finite-memory strategies.

\noindent
Second, we show that under the hypothesis that both players play either 
finite-memory or memoryless strategies, the generalized mean-payoff game 
and the generalized energy game problems are equivalent.

\noindent
Third, our main contribution is the study of the computational complexity of the 
decision problems for generalized mean-payoff games and generalized energy games, 
both for finite-memory strategies and the special case of memoryless strategies. 
Our complexity results can be summarized as follows:
(A)~For finite-memory strategies, we provide a nondeterministic polynomial time 
algorithm 
for deciding negative instances of the problems\footnote{Negative  
instances are those where Player~1 is losing, and by determinacy under finite-memory where Player~2 is winning.}.
Thus we show that the decision problems are in coNP. 
This significantly improves the complexity as compared to the EXPSPACE 
algorithm that can be obtained by reduction to {\sc Vass} (vector addition systems with states)~\cite{BJK10}.
Furthermore, we establish a coNP lower bound for these problems by reduction 
from the complement of the  3SAT problem, hence showing that the problem
is coNP-complete.
(B)~For the case of memoryless strategies, as the games are not determined, we 
consider the problem of determining if Player 1 has a memoryless 
winning strategy. First, we show that the problem of determining if Player 1 
has a memoryless winning strategy is in NP, 
and then show that the problem is NP-hard (i)~even when the weights are restricted 
to $\{-1,0,1\}$; or
(ii)~when the weights are arbitrary and the dimension is~2.

\noindent{\em Related works.}
Mean-payoff games, which are the one-dimension version of our generalized 
mean-payoff games, have been extensively studied starting with the works of 
Ehrenfeucht and Mycielski in~\cite{EM79} where they prove memoryless determinacy 
for these games. Because of memoryless determinacy, it is easy to show that 
the decision problem for mean-payoff games lies in NP~$\cap$~coNP, 
but despite large research efforts, no polynomial time algorithm is known for 
that problem. A pseudo-polynomial time algorithm has been proposed by Zwick 
and Paterson in~\cite{ZP96}, and improved in~\cite{BCDGR09}. The one-dimension
special case of generalized energy games have been introduced in~\cite{CAHS03}
and further studied in~\cite{BFLMS08} where log-space equivalence with 
classical mean-payoff games is established.




Generalized energy games can be viewed as games played on {\sc Vass} (vector addition systems with states)
where the objective
is to avoid unbounded decreasing of the counters. A solution to such games 
on {\sc Vass} is provided in~\cite{BJK10} (see in particular Lemma 3.4 in~\cite{BJK10}) with a PSPACE 
algorithm when the weights are $\{-1,0,1\}$, leading to an EXPSPACE algorithm when the
weights are arbitrary integers.
We drastically improve the EXPSPACE upper-bound by providing a coNP
algorithm for the problem, and we also provide a coNP lower bound even when 
the weights are restricted to $\{-1,0,1\}$.

\section{Generalized Mean-payoff and Energy Games}\label{sec:def}

\noindent{\bf Well quasi-orders.} 
Let $D$ be a set. A relation $\preceq$ over $D$ is a {\em well quasi-order}, 
wqo for short, if the following holds: (a)~$\preceq$ is transitive and reflexive;
and (b)~for all $f : \mathbb{N} \rightarrow D$, there exists 
$i_1,i_2 \in \mathbb{N}$ such that $i_1 < i_2$ and $f(i_1) \preceq f(i_2)$.

\begin{lemma}\label{wqo}
$(\mathbb{N}^k,\leq)$ is well quasi-ordered. 
\end{lemma} 

\noindent{\bf Multi-weigthed two-player game structures.}
A {\em multi-weigthed two-player game structure} is a tuple 
$G=(S_1,S_2,s_{{\sf init}},E,k,w)$ where $S_1 \cap S_2 = \emptyset$, 
and $S_i$ ($i = 1,2$) is the finite set of {\em Player~$i$ positions}, 
$s_{{\sf init}} \in S_1$ is the {\em initial position}, 
$E \subseteq (S_1  \cup S_2) \times (S_1 \cup S_2)$ is 
the set of {\em edges} such that for all $s \in S_1 \cup S_2$, 
there exists $s' \in S_1 \cup S_2$ such that $(s,s') \in E$, $k \in \nat$ is 
the {\em dimension} of the multi-weights, $w : E \rightarrow \integ^k$ is the 
{\em multi-weight labeling function}. 
$G$ is a multi-weighted {\em one-player} game structure if $S_2 = \emptyset$.

A {\em play} in $G$ is an infinite sequence of $\pi=s_0 s_1 \dots s_n \dots$ 
such that $(i)$ $s_0=s_{{\sf init}}$, $(ii)$ for all $i \geq 0$ we have $(s_i,s_{i+1}) \in E$. 
A play $\pi=s_0 s_1 \dots s_n \dots$ is {\em ultimately periodic} if 
it can be decomposed as $\pi=\rho_1 \cdot \rho_2^{\omega}$ where $\rho_1$ and 
$\rho_2$ are two finite sequences of positions. The {\em prefix} up to position
$n$ of a play $\pi=s_0 s_1 \dots s_n \dots$ is the finite sequence 
$\pi(n)=s_0 s_1 \dots s_n$, its last element $s_n$ is denoted by $\Last(\pi(n))$. 
A prefix $\pi(n)$ belongs to Player~$i$ ($i \in \{1,2\}$) if $\Last(\pi(n)) \in S_i$. 
The set of plays in $G$ is denoted by $\Plays(G)$, 
the corresponding set of prefixes is denoted by $\Prefs(G)$, 
the set of prefixes that belongs to Player~$i$ ($i \in \{1,2\}$) is denoted 
by $\Prefs_i(G)$, and the set of ultimately periodic plays in $G$ is denoted by $\Plays^{up}(G)$. 

The {\em energy level vector} of a prefix of play $\rho=s_0 s_1 \dots s_n$ is 
$\EL(\rho)=\sum_{i=0}^{i=n-1} w(s_i,s_{i+1})$, and the {\em mean-payoff vector} 
of an  ultimately periodic play 
$\pi=s_0 s_1 \dots s_n \dots$ is $\MP(\pi)=\lim_{n\rightarrow \infty} \frac{1}{n} \EL(\pi(n))$.

\smallskip\noindent{\bf Strategies.}
A strategy for Player~$i$ ($i \in \{1,2\}$) in~$G$ is a function 
$\lambda_i : \Prefs_i(G) \rightarrow S_1 \cup S_{2}$ such that for all $\rho \in \Prefs_i(G)$ we have
$(\Last(\rho),\lambda_i(\rho)) \in E$. A play $\pi=s_0 s_1 \dots \in \Plays(G)$
is \emph{consistent} with a strategy $\lambda_i$ of Player~$i$ if 
$s_{j+1}=\lambda_i(s_0 s_1 \dots s_j)$ for all $j \geq 0$ such that $s_j \in S_i$.
The {\em outcome of a pair of strategies}, $\lambda_1$ for Player~1 and $\lambda_2$ 
for Player~2, is the (unique) play which is consistent with both $\lambda_1$ and $\lambda_2$.
We denote $\outcome_G(\lambda_1,\lambda_2)$ this outcome. 


A strategy $\lambda_1$ for Player~$1$ has {\em finite-memory} if it can be encoded 
by a deterministic Moore machine $(M,m_0,\alpha_u,\alpha_n)$ where $M$ is a finite 
set of states (the memory of the strategy), $m_0 \in M$ is the initial memory state, 
$\alpha_u: M \times (S_{1} \cup S_2) \to M$ is an update function, 
and $\alpha_n : M \times S_{i} \to S_{1} \cup S_2$ is the next-action function. 
If the game is in a Player-$1$ position $s \in S_1$ and $m \in M$ is the current memory value,
then the strategy chooses $s' = \alpha_n(m,s)$ as the next position and the memory 
is updated to $\alpha_u(m,s)$. Formally, $\tuple{M, m_0, \alpha_u, \alpha_n}$
defines the strategy $\lambda$ such that $\lambda(\rho\cdot s) = \alpha_n(\hat{\alpha}_u(m_0, \rho), s)$
for all $\rho \in (S_1 \cup S_2)^*$ and $s \in S_1$, where $\hat{\alpha}_u$ extends $\alpha_u$ to sequences
of positions as expected. A strategy is \emph{memoryless} if $\abs{M} = 1$.
For a finite-memory strategy $\lambda_1$ of Player~1, let $G_{\lambda_1}$ be the graph obtained as the product
of $G$ with the Moore machine defining $\lambda_1$, with initial vertex $\tuple{m_0,s_{{\sf init}}}$ 
and where $(\tuple{m,s},\tuple{m',s'})$ is a transition in $G_{\lambda_1}$ if 
$m' = \alpha_u(m,s)$, and either $s \in S_1$ and $s'=\alpha_n(m,s)$, or $s \in S_2$ and $(s,s') \in E$.
The set of inifinite paths in $G_{\lambda_1}$ and the set of plays consistent with
$\lambda_1$ coincide. A similar definition can be given for the case of Player~2.


\smallskip\noindent{\bf Objectives.}
An {\em objective} for Player~$1$ in $G$ is a set of plays $W \subseteq \Plays(G)$. 
A strategy $\lambda_1$ for Player~1 is {\em winning} for $W$ in $G$ if for all plays in 
$\pi \in \Plays(G)$ that are consistent with $\lambda_1$, we have that $\pi \in W$. 
A strategy $\lambda_2$ for Player~2 is {\em spoiling} for $W$ in $G$ if for all plays in 
$\pi \in \Plays(G)$ that are consistent with $\lambda_2$, we have that $\pi \not\in W$. 
We consider the following objectives:
 \begin{itemize}
 	\item {\em Multi Energy objectives}. Given an initial energy vector 
$v_0 \in \nat^k$, the {\em multi energy objective} 
$\PosEnergy_G(v_0)=\{ \pi \in \Plays(G) \mid \forall n \geq 0 : v_0 + \EL(\pi(n)) \in \nat^k \}$ 
requires that the energy level in all dimensions remains always nonnegative.
	\item {\em Multi Mean-payoff objectives}. Given a threshold vector 
$v \in \mathbb{Z}^k$, the {\em multi mean-payoff objective} 
$\MeanPayoff_G(v)=\{ \pi \in \Plays^{up}(G) \mid  \MP(\pi) \geq v \}$ 
requires for all dimensions $j$ the mean-payoff for dimension $j$ is at least 
$v(j)$.
 \end{itemize}

\smallskip\noindent{\bf Decision problems.}
We consider the following decision problems:
\begin{itemize}
\item The {\em unknown initial credit problem} asks, given an 
multi-weighted two-player game structure $G$, to decide whether there exists
an initial credit vector $v_0 \in \nat^k$ and a winning strategy $\lambda_1$ 
for Player~1 for the objective  $\PosEnergy_G(v_0)$.
\item The {\em mean-payoff threshold problem} (for finite memory) asks, given an multi-weighted two-player 
game structure $G$ and a threshold vector $v \in \mathbb{Z}^k$, to decide whether there 
exists a \emph{finite-memory} strategy $\lambda_1$ for Player~1 such that for all 
\emph{finite-memory} strategies 
$\lambda_2$ of Player~2, $\outcome_G(\lambda_1,\lambda_2) \in \MeanPayoff_G(v)$.
 \end{itemize}

\noindent Note that in the unknown initial credit problem, we allow arbitrary strategies
(and we show in Theorem~\ref{thrm_gen_energy_fin} that actually finite-memory strategies  are sufficient),
while in the mean-payoff threshold problem, we require finite-memory strategy
which is restriction (according to Theorem~\ref{thrm_gen_mean}) of a more general problem 
of deciding the existence of arbitrary winning strategies.

\smallskip\noindent{\bf Determinacy and determinacy under finite-memory.} 
A game $G$ with an objective $W$ is \emph{determined} if either Player~1 has
a winning strategy, or Player~2 has a spoiling strategy.
A game $G$ with an objective $W$ is \emph{determined under finite-memory} if either 
(a)~Player~1 has a \emph{finite-memory} strategy $\lambda_1$ such that for all 
\emph{finite-memory} strategies $\lambda_2$ of Player~2, 
we have $\outcome_G(\lambda_1,\lambda_2) \in W$; or
(b)~Player~2 has a \emph{finite-memory} strategy $\lambda_2$ such that for all 
\emph{finite-memory} strategies $\lambda_1$ of Player~1, 
we have $\outcome_G(\lambda_1,\lambda_2) \not\in W$.
Games with objectives $W$ are determined (resp. determined under finite-memory) if 
all game structures with objectives $W$ are determined (resp. determined under finite-memory).
We say that determinacy and determined under finite-memory coincide for a 
class of objectives, 
if for all objectives in the class and all game structures,  
the answer of the determinacy and determined under finite-memory 
coincide (i.e., Player~1 has a winning strategy iff there is a finite-memory 
winning strategy, and similarly for Player~2). 
Generalized mean-payoff and energy objectives are measurable: 
(a)~generalized mean-payoff objectives can be expressed as finite intersection 
of mean-payoff objectives and mean-payoff objectives are complete 
for the third level of Borel hierarchy~\cite{ChaTCS}; and 
(b) generalized energy objectives can be expressed as finite intersection 
of energy objectives, and enery objectives are closed sets.   
Hence determinacy of generalized mean-payoff and energy games follows 
from the result of~\cite{Mar75}.

\begin{theorem}[Determinacy~\cite{Mar75}]
Generalized mean-payoff and energy games are determined. 
\end{theorem}

\section{Determinacy under Finite-memory and Inter-reducibility}

In this section, we establish four results. First, we show that to win generalized energy games, it is sufficient for Player 1 to 
play {\em finite-memory strategies}. Second, we show that to spoil generalized energy games, it is sufficient for Player 2 
to play {\em memoryless strategies}. As a consequence, generalized energy games are determined under finite-memory. 
Third, using this finite-memory determinacy result, we show that the decision problems for generalized energy and 
mean-payoff games (see Section~\ref{sec:def}) are log-space inter-reducible. Finally, 
we show that infinite-memory strategies are more powerful than finite-memory strategies in generalized mean-payoff games.

For generalized energy games, we first show that finite-memory strategies are sufficient for Player~1,
and then that memoryless strategies are sufficient for Player~2.

\begin{lemma}
\label{lem:energy-player1-finite-memory}
For all multi-weighted two-player game structures $G$, the answer to the unknown 
initial credit problem is {\sc Yes} iff there exists a initial credit 
$v_0 \in \mathbb{N}^k$ and a finite-memory strategy $\lambda^{\sf FM}_1$ 
for Player 1 such that for all strategies $\lambda_2$ of Player 2, 
$\outcome_G(\lambda^{\sf FM}_1,\lambda_2) \in \PosEnergy_G(v_0)$.
\end{lemma}
\begin{proof}  
One direction is trivial. For the other direction, assume that $\lambda_1$ is a 
(not necessary finite-memory) winning strategy for Player~1 in $G$ with initial 
credit $v_0 \in \mathbb{N}^k$. We show how to construct from $\lambda_1$ 
a finite-memory strategy $\lambda_1^{\sf FM}$ which is winning  against all 
strategies of Player~2 for initial credit $v_0$. For that we consider the 
unfolding of the game graph $G$ in which Player~1 plays according to $\lambda_1$. 
This infinite tree, noted $T_{G(\lambda_1)}$, has as set of nodes all the 
prefixes of plays in $G$ when Player~1 plays according to $\lambda_1$. 
We associate to each node $\rho=s_0 s_1 \dots s_n$ in this tree the energy 
vector $v_0+\EL(\rho)$. As $\lambda_1$ is winning, we have that 
$v_0+\EL(\rho) \in \mathbb{N}^k$ for all $\rho$. Now, consider the set 
$(S_1 \cup S_2) \times \mathbb{N}^k$, and the relation $\sqsubseteq$ on this set 
defined as follows: $(s_1,v_1) \sqsubseteq (s_2,v_2)$ iff $s_1=s_2$ and 
$v_1 \leq v_2$ i.e., for all $i$, $1 \leq i \leq k$, $v_1(i) \leq v_2(i)$. 
The relation $\sqsubseteq$ is a wqo (easy consequence of Lemma~\ref{wqo}). 
As a consequence, on every infinite branch $\pi=s_0 s_1 \dots s_n \dots$ 
of $T_{G(\lambda_1)}$, there exists two positions $i < j$ such that 
$\Last(\pi(i))=\Last(\pi(j))$ and $\EL(\pi(i)) \leq \EL(\pi(j))$.  
We say that node $j$ subsumes node $i$. Now, let $T^{\sf FM}_{G(\lambda_1)}$ 
be the tree $T_{G(\lambda_1)}$ where we stop each branch when we reach a node 
$n_2$ which subsumes one of its ancestor node $n_1$. 
Clearly, $T^{\sf FM}_{G(\lambda_1)}$ is finite. 
Also, it is easy to see that Player~1 can play in the subtree rooted at $n_2$ as 
she plays in the subtree rooted in $n_1$ because its energy level in $n_2$ is 
greater than in $n_1$.  From $T^{\sf FM}_{G(\lambda_1)}$, we can construct a 
Moore machine which encode a finite-memory strategy $\lambda_1^{\sf FM}$ which 
is winning the generalized energy game $G$ as it is winning for initial energy 
level $v_0$.
\end{proof}

\begin{lemma}\cite{BJK10}
\label{lem:player-two-memoryless}
For all multi-weigthed two-player game structures $G$, the answer to the 
unknown initial credit problem is {\sc No} if and only if there exists
a \emph{memoryless} strategy $\lambda_2$ for Player~$2$, such that
for all initial credit vectors $v_0 \in \nat^k$ and all strategies $\lambda_1$ 
for Player~$1$ we have $\outcome_G(\lambda_1,\lambda_2) \not\in \PosEnergy_G(v_0)$.
\end{lemma}

\begin{proof}  
The proof was given in~\cite{BJK10}[Lemma 19]. Intuitively, consider a Player-$2$ 
state $s \in S_2$ with two sucessors $s',s''$. 
If an initial credit vector $v'_0$ is sufficient for Player~$1$ to win
against Player-$2$ always choosing $s'$, and $v''_0$ is sufficient against  
Player-$2$ always choosing $s''$, then $v'_0+v''_0$ is sufficient against 
Player-$2$ arbitrarily alternating between $s'$ and $s''$. This is because
of the fact that if Player~$1$ maintains all energies nonnegative when initial credit
is $v_0$, then he can maintain all energies above $\Delta$ when initial credit
is $v_0 + \Delta$ ($\Delta \in \nat^k$).
\end{proof}

\noindent As a consequence of the two previous lemmas, we have the following theorem.

\begin{theorem}\label{thrm_gen_energy_fin}
Generalized energy games are determined under finite-memory, and determinacy 
coincide with determinacy under finite-memory for generalized energy games.
\end{theorem}

\begin{remark}
Note that even if Player 2 can be restricted to play memoryless strategies in generalized energy games, it may be that Player~$1$ is winning with some initial 
credit vector $v_0$ when Player~$2$ is memoryless, and is not winning with 
the same initial credit vector $v_0$ when Player~$2$ can use arbitrary strategies.
This situation is illustrated in \figurename~\ref{fig:memory-needed} where Player~$1$
(owning round states) can maintain the energy nonegative in all dimensions 
with initial credit $(2,0)$ when Player~$2$ (owning square states) is memoryless.
Indeed, either Player~$2$ chooses the left edge from $q_0$ to $q_1$ and Player~$1$ wins, 
or Player~$2$ chooses the right edge from $q_0$ to $q_2$, and Player~$1$ wins as well by 
alternating the edges back to $q_0$. Now, if Player~$2$ has memory, then Player~2 wins
by choosing first the right edge to $q_2$, which forces Player~$1$ to come back
to $q_0$ with multi-weight $(-1,1)$. The energy level is now $(1,1)$ in $q_0$ and Player~$2$
chooses the left edge to $q_1$ which is losing for Player~$1$. Note that Player~$1$
wins with initial credit $(2,1)$ and $(3,0)$ (or any larger credit) against all 
arbitrary strategies of Player~$2$.
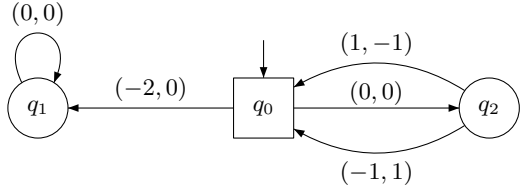
\begin{figure}[!tb]
 \begin{center}
   \begin{picture}(75,28)(0,0)


\node[Nmarks=i, iangle=90, Nmr=0](n0)(40,12){$q_0$}
\node[Nmarks=n](n1)(10,12){$q_1$}
\node[Nmarks=n](n2)(70,12){$q_2$}

\drawloop[ELside=l,loopCW=y, loopdiam=6](n1){$(0,0)$}


\drawedge[ELpos=50, ELside=l, ELdist=0.5, curvedepth=0](n0,n2){$(0,0)$}
\drawedge[ELpos=50, ELside=l, curvedepth=6](n2,n0){$(-1,1)$}
\drawedge[ELpos=50, ELside=r, curvedepth=-6](n2,n0){$(1,-1)$}

\drawedge[ELpos=50, ELside=r, curvedepth=0](n0,n1){$(-2,0)$}


\end{picture}
   \caption{Player~$1$ (round states) wins with initial credit $(2,0)$ when Player~$2$ (square states) can use memoryless strategies, 
but not when Player~$2$ can use arbitrary strategies.   \label{fig:memory-needed}}
 \end{center}
\end{figure}
\end{remark}

We now  show that generalized mean-payoff games (where players are restricted to play finite-memory strategies by definition) 
are log-space equivalent to generalized energy games. 
First note that the mean-payoff threshold problem with threshold vector $v \in \integ^k$ can 
be reduced to the mean-payoff threshold problem with threshold vector $\{0\}^k$,
by shifting all multi-weights in the game graph by $v$ (which has the effect of 
shifting the mean-payoff value by $v$). Given this reduction, the following result
shows that the unknown initial credit problem (for multi-energy games) and
the mean-payoff threshold problem (with finite-memory strategies) are equivalent.

\begin{theorem}\label{thrm_inter}
\label{lem:energy-mean-payoff-reduction}
For all multi-weigthed two-player game structures $G$ with dimension $k$, the answer to the 
unknown initial credit problem is {\sc Yes} if and only if the answer to the 
mean-payoff threshold problem (for finite memory) with threshold vector $\{0\}^k$ is {\sc Yes}.
\end{theorem}
\begin{proof}  
First, assume that there exists a winning strategy $\lambda_1$ for Player~$1$
in~$G$ for the multi energy objective $\PosEnergy_G(v_0)$ (for some $v_0$). 
Theorem~\ref{thrm_gen_energy_fin} establishes that finite memory is sufficient to win
multi-energy games, so we can assume that $\lambda_1$ has finite memory.
Consider the restriction of the graph $G_{\lambda_1}$ to the reachable vertices,
and we show that the energy vector of every simple cycle is nonnegative. By contradcition,
if there exists a simple cycle with energy vector negative in one dimension,
then the infinite path that reaches this cycle and loops through it forever
would violate the objective $\PosEnergy_G(v_0)$ regardless of the vector $v_0$.

Now, this shows that every reachable cycle in $G_{\lambda_1}$ has nonnegative
mean-payoff value in all dimensions, hence $\lambda_1$ is winning for the 
multi mean-payoff objective $\MeanPayoff_G(\{0\}^k)$.

Second, assume that there exists a finite-memory strategy $\lambda_1$ for Player~$1$
that is winning in~$G$ for the multi mean-payoff objective $\MeanPayoff_G(\{0\}^k)$. 
By the same argument as above, all simple cycles in $G_{\lambda_1}$ are nonnegative 
and the strategy $\lambda_1$ is also winning for the objective $\PosEnergy_G(v_0)$ 
for some $v_0$. Taking $v_0 = \{n W\}^k$ where $n$ is the number of states 
in~$G_{\lambda_1}$ (which bounds the length of the acyclic paths) and $W \in \integers$ 
is the largest weight in the game suffices.
\end{proof}

Note that the result of Theorem~\ref{thrm_inter} does not hold for arbitrary
strategies as shown in the following lemma.

\begin{lemma}\label{lemm_inf_power}
In generalized mean-payoff games, infinite memory may be necessary to win
(finite-memory strategies may not be sufficient). 
\end{lemma}

\begin{proof} 
To show this, we first need to define
the mean-payoff vector of arbitrary plays (because arbitrary strategies, i.e.,
infinite-memory strategies, may produce non-ultimately periodic plays). In particular, the limit
of $\frac{1}{n} \cdot \EL(\pi(n))$ for $n \to \infty$ may not exist for arbitrary plays~$\pi$.
Therefore, two possible definitions are usually considered, namely either
$\underline{\MP}(\pi) = \liminf_{n \to \infty} \frac{1}{n} \cdot \EL(\pi(n))$,
or $\overline{\MP}(\pi) = \limsup_{n \to \infty} \frac{1}{n} \cdot \EL(\pi(n))$.
In both cases, better payoff can be obtained with infinite memory:
the example of \figurename~\ref{fig:crazy} shows a game where all states belong to
Player~$1$. We claim that $(a)$ for $\underline{\MP}$, Player~$1$ can achieve
a threshold vector $(1,1)$, and $(b)$ for $\overline{\MP}$,
Player~$1$ can achieve a threshold vector $(2,2)$; $(c)$ if we restrict Player~$1$
to use a finite-memory strategy, then it is not possible to win the
multi mean-payoff objective with threshold $(1,1)$
(and thus also not with $(2,2)$). To prove $(a)$, consider the strategy
that visits $n$ times $q_a$ and then $n$ times $q_b$, and repeats this forever
with increasing value of $n$. This guarantees a mean-payoff
vector $(1,1)$ for $\underline{\MP}$ because in the long-run roughly half of the 
time is spent in $q_a$ and roughly half of the time in $q_b$.
To prove~$(b)$, consider the strategy that
alternates visits to $q_a$ and $q_b$ such that after the $n$th alternation,
the self-loop on the visited state $q$ ($q \in \{q_a,q_b\}$) is taken so
many times that the average frequency of $q$ gets larger than~$\frac{1}{n}$
in the current finite prefix of the play.
This is always possible and achieves threshold $(2,2)$ for $\overline{\MP}$.
Note that the above two strategies require infinite memory. To prove $(c)$,
notice that finite-memory strategies produce an ultimately periodic play
and therefore $\underline{\MP}$ and $\overline{\MP}$ coincide with $\MP$. It is easy to
see that such a play cannot achieve $(1,1)$ because the periodic
part would have to visit both $q_a$ and $q_b$ and then the mean-payoff vector $(v_1,v_2)$
of the play would be such that $v_1 + v_2 < 2$ and thus $v_1 = v_2 = 1$ is
impossible.
\end{proof}

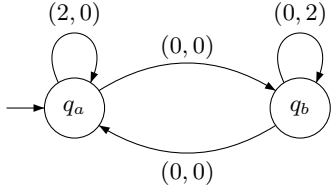
\begin{figure}[!tb]
 \begin{center}
   \begin{picture}(48,28)(0,0)


\node[Nmarks=i, iangle=180](n0)(10,12){$q_a$}
\node[Nmarks=n](n1)(40,12){$q_b$}

\drawloop[ELside=l,loopCW=y, loopdiam=6](n0){$(2,0)$}
\drawloop[ELside=l,loopCW=y, loopdiam=6](n1){$(0,2)$}


\drawedge[ELpos=50, ELside=l, ELdist=0.5, curvedepth=6](n0,n1){$(0,0)$}
\drawedge[ELpos=50, ELside=l, curvedepth=6](n1,n0){$(0,0)$}


\end{picture}
   \caption{A generalized mean-payoff game where infinite memory is necessary to win (Lemma~\ref{lemm_inf_power}).}\label{fig:crazy}
 \end{center}
\end{figure}

\noindent Theorem~\ref{thrm_inter} and Lemma~\ref{lemm_inf_power}, 
along with Theorem~\ref{thrm_gen_energy_fin} gives the following result.

\begin{theorem}\label{thrm_gen_mean}
Generalized mean-payoff games are determined under finite-memory, however 
determinacy and determined under finite-memory do not coincide for generalized mean-payoff 
games.
\end{theorem}

\section{coNP-completeness for Finite-Memory Strategies}

In this section, we present a nondeterministic polynomial time algorithm 
to recognize the instances for which there is no winning strategies for Player~1 
in a multi-energy game. 
First, we show that the one-player version of this game can be solved 
by checking the existence of a circuit (i.e., a not necessarily simple cycle) 
with overall nonnegative effect 
in all dimensions. Second, we build on this and the memoryless result
for Player~2 to define a coNP algorithm. 
The main result (Theorem~\ref{thrm_complete}) is derived from 
Lemma~\ref{lem:coNp-membership} and Lemma~\ref{thrm_hard} below.

\begin{theorem}\label{thrm_complete}
The unknown initial credit and the mean-payoff threshold problems for 
multi-weighted two-player game structures are coNP-complete. 
\end{theorem}

\smallskip\noindent{\bf coNP upper bound.}
First, we need the following result about finding zero circuits in 
multi-weighted directed graphs (a graph is a one-player game).
A zero circuit is a finite sequence $s_0 s_1 \dots s_n$ such that
$s_0 = s_n$, $(s_i,s_{i+1}) \in E$ for all $0 \leq i < n$, and 
$\sum_{i=0}^{n-1} w(s_i,s_{i+1}) = (0,0,\dots,0)$. The circuit need
not be simple.

\begin{lemma}[\cite{KS88}]\label{lem:zero-cycle}
Determining if a $k$-dimensional directed graph contains a zero circuit
can be done in polynomial time.
\end{lemma}


\begin{lemma}\label{lem:coNp-membership}
The unknown initial credit and the mean-payoff threshold problems for 
multi-weighted two-player game structures are in coNP. 
\end{lemma}

\begin{proof}
By Lemma~\ref{lem:player-two-memoryless}, we know that Player~2 can be restricted 
to play memoryless strategies. A coNP algorithm can guess a memoryless 
strategy $\lambda$ and check in polynomial time that it is winning using the following argument.

First, consider the graph $G_{\lambda}$ as a one-player game (in which all states
belong to player~$1$. We show that if there exists an initial energy level $v_0$ and an 
infinite play $\pi=s_0 s_1 \dots s_n \dots$ in $G_{\lambda}$ such that $\pi \in \PosEnergy(v_0)$ 
then there exist a reachable circuit in $G_{\lambda}$ that has nonnegative effect in all 
dimensions. To show that, we extend $\pi$ with the energy information as follows: 
$\pi'=(s_0,w_0) (s_1,w_1) \dots (s_n,w_n) \dots$ where $w_0=v_0$ and for all 
$i \geq 1$, $w_i=v_0+\EL(\pi(i))$. As $\pi \in \PosEnergy(v_0)$, we know that 
for all $i \geq 0$, $w_i \in \mathbb{N}^k$. So, we can define the following order 
on the pairs $(s,w) \in (S_1 \cup S_2) \times \mathbb{N}^k$ in the run: 
$(s,w) \sqsubseteq (s',w')$ iff $s=s'$ and $w(j) \leq w'(j)$ for all $1 \leq j \leq k$.  
From Lemma~\ref{wqo}, it is easy to show that $\sqsubseteq$ is a wqo. 
Then  there exist two positions $i_1 < i_2$ in $\pi'$ such that 
$(s_{i_1},w_{i_1}) \sqsubseteq (s_{i_2},w_{i_2})$. 
The circuit underlying those two positions has nonnegative effect in all dimensions.

Based on this, we can decide if there exists an initial energy vector $v_0$ and an infinite path 
in $G_{\lambda}$ that satisfies $\PosEnergy_G(v_0)$ using the result of Lemma~\ref{lem:zero-cycle} 
on modified version of $G_{\lambda}$ obtained as follows. 
In every state of $G_{\lambda}$, we add $k$ self-loops with respective multi-weight $(-1,0,\dots,0)$,
$(0,-1,0,\dots,0)$, $\dots$, $(0,\dots,0,-1)$, i.e. each self-loop removes one unit
of energy in one dimension. It is easy to see that $G_{\lambda}$ has a circuit with nonnegative 
effect in all dimensions if and only if the modified $G_{\lambda}$ has a zero circuit, which 
can be determined in polynomial time. The result follows.
\end{proof}

\smallskip\noindent{\bf Lower bound: coNP-hardness.}
We show that the unknown initial credit problem for 
multi-weighted two-player game structures is coNP-hard. 
We present a reduction from the complement of the 3SAT problem 
which is NP-complete~\cite{PapaBook}.

\smallskip\noindent\emph{Hardness proof.}
We show that the problem of deciding whether Player~1 has a winning 
strategy for the unknown initial credit problem for multi-weighted 
two-player game structures is at least as hard as 
deciding whether a 3SAT formula is unsatisfiable.
Consider a 3SAT formula $\psi$ in CNF with clauses $C_1,C_2,\ldots,C_k$ 
over variables $\{x_1, x_2, \ldots, x_n\}$, where each clause consists 
of disjunctions of exactly three literals (a literal is a variable or its
complement). 
Given the formula $\psi$, we construct a game graph as shown in 
Figure~\ref{fig:3sat}.  
The game graph is as follows: from the initial position, Player~1 chooses 
a clause, then from a clause Player~2 chooses a literal  that appears 
in the clause (i.e., makes the clause true).  From every literal 
the next position is the initial position.
We now describe the multi-weight labeling function $w$.
In the multi-weight function there is a component for every literal.
For edges from the initial position to the clause positions, and from 
the clause positions to the literals, the weight for every component 
is~0. 
We now define the weight function for the edges from literals back to the 
initial position: for a literal $y$, and the edge from $y$ to the 
initial position, the weight for the component of $y$ is~$1$, the weight for
the component of the complement of $y$ is~$-1$, and for all the other 
components the weight is~$0$.
We now define a few notations related to assignments of truth values 
to literals.
We consider \emph{assignments} that assign truth values to all the literals.
An assignment is \emph{valid} if for every literal the truth value assigned
to the literal and its complement are complementary (i.e., for all $1 \leq i 
\leq n$, if 
$x_i$ is assigned true (resp. false), then the complement $\overline{x}_i$ 
of $x_i$ is assigned false (resp. true)).
An assignment that is not valid is \emph{conflicting} (i.e., for some 
$1 \leq i \leq n$, both $x_i$ and $\overline{x}_i$ are assigned the same 
truth value). 
If the formula $\psi$ is satisfiable, then there is a valid assignment 
that satisfies all the clauses.
If the formula $\psi$ is not satisfiable, then every assignment that satisfies
all the clauses must be conflicting. 
We now present two directions of the hardness proof.

\begin{figure}[tb]
 \centering
 \begin{picture}(65,40)(0,0)

   \drawcurve(51,33)(47,44)(10,20)
   \drawcurve(51,20)(49,12)(15,15)
   \drawcurve(51,5)(49,-2)(10,10)

   \node[Nmarks=i](q0)(10,15){}
   \node[Nmr=0](q1)(35,33){$C_1$}
   \fmark[fangle=20,flength=10](q1)
   \fmark[fangle=0,flength=10](q1)
   \fmark[fangle=-20,flength=10](q1) 
   \node[Nmr=0](q2)(35,20){$C_2$}
   \fmark[fangle=20,flength=10](q2)
   \fmark[fangle=0,flength=10](q2)
   \fmark[fangle=-20,flength=10](q2) 
   \node[Nmr=0](q3)(35,5){$C_k$}
   \fmark[fangle=20,flength=10](q3)
   \fmark[fangle=0,flength=10](q3)
   \fmark[fangle=-20,flength=10](q3) 

   \drawedge(q0,q1){ }
   \drawedge(q0,q2){ }
   \drawedge(q0,q3){ }

   \put(35,11){$\vdots$}
   \put(53,31){{\Huge $\}$}}
   \put(53,18){{\Huge $\}$}}
   \put(53,3){{\Huge $\}$}}
   \put(58,32){ literal}
   \put(58,19){ literal}
   \put(58,4){ literal}

 \end{picture}
 \caption{Game graph construction for a 3SAT formula (Lemma~\ref{thrm_hard}).}
 \label{fig:3sat}
\end{figure}
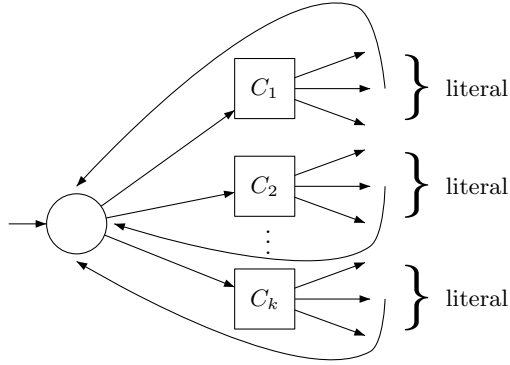

\smallskip\noindent\emph{$\psi$ satisfiable implies Player~2 winning.}
We show that if $\psi$ is satisfiable, then Player~2 has a memoryless
winning strategy.  
Since $\psi$ is satisfiable, there is a valid assignment $A$ that 
satisfies every clause.
The memoryless strategy is constructed from the assignment $A$ as follows:
for a clause $C_i$, the strategy chooses a literal as successor that appears 
in $C_i$ and is set to true by the assignment. 
Consider an arbitrary strategy for Player~1, and the infinite play: 
the literals visited in the play are all assigned truth values true by $A$,
and the infinite play must visit some literal infinitely often. 
Consider the literal $x$ that appears infinitely often in the play, then the 
complement literal $\overline{x}$ is never visited, and every 
time literal $x$ is visited, the component corresponding to $\overline{x}$ 
decreases by~$1$, and since $x$ appears infinitely often it follows that
the play is winning for Player~2 for every finite initial credit.  
It follows that the strategy for Player~2 is winning, and the answer to the 
unknown initial credit problem is ``No".

\smallskip\noindent\emph{$\psi$ not satisfiable implies Player~1 is winning.} 
We now show that
if $\psi$ is not satisfiable, then Player~1 is winning. 
By determinacy, it suffices to show that Player~2 is not winning, 
and by existence of memoryless winning strategy for Player~2 
(Lemma~\ref{lem:player-two-memoryless}), it suffices to show that there is no 
memoryless winning strategy for Player~2. 
Fix an arbitrary memoryless strategy for Player~2, (i.e., in every clause
Player~2 chooses a literal that appears in the clause). 
If we consider the assignment $A$ obtained from the memoryless strategy, then
since $\psi$ is not satisfiable it follows that the assignment $A$ is 
conflicting.
Hence there must exist clause $C_i$ and $C_j$ and variable $x_k$ such that the 
strategy chooses the literal $x_k$ in $C_i$ and the complement variable 
$\overline{x}_k$ in $C_j$. 
The strategy for Player~1 that at the starting position alternates between 
clause $C_i$ and $C_j$, along with that the initial credit of $1$ for the 
component of $x_k$ and $\overline{x}_k$, and~$0$ for all other components, 
ensures that the strategy for Player~2 is not winning. 
Hence the answer to the unknown initial credit problem is ``Yes", and we have 
the following result.

\begin{lemma}\label{thrm_hard}
The unknown initial credit and the mean-payoff threshold problems for 
multi-weighted two-player game structures are coNP-hard. 
\end{lemma}

\noindent Observe that our hardness proof works with weights restricted to the 
set $\{-1,0,1\}$.

\section{NP-completeness for Memoryless Strategies}
In this section we consider the unknown initial credit and the mean-payoff 
threshold problems for multi-weighted two-player game structures when 
Player~1 is restricted to use memoryless strategies.
We will show NP-completeness for these problems.

\begin{lemma}\label{lemm_memless_1}
The unknown intial credit and the mean-payoff threshold problems for 
multi-weighted two-player game structures for memoryless strategies for 
Player~1 lie in NP.
\end{lemma}

\begin{proof}  
The inclusion in NP is obtained as follows: the polynomial witness is the
memoryless strategy for Player~1, and once the strategy is fixed we obtain 
a game graph with choices for Player~2 only. 
The verification problem for the unknown initial credit checks that for every 
dimension there is no negative cycle, and the verification problem for 
mean-payoff threshold checks that for every dimension every cycle satisfy the
threshold condition. 
Both the above verification problem can be achieved in polynomial time by
solving the energy-game and mean-payoff game problem on graphs with choices 
for Player~2 only~\cite{Karp,BFLMS08,CAHS03}.
The desired result follows.
\end{proof}

Lemma~\ref{lemm_memless_2} shows NP-hardness for dimension $k=2$ 
and arbitrary integral weights, and is obtained by a reduction 
from the {\sc Knapsack} problem.
If $k=1$, then the problems reduces to the classical energy and mean-payoff games,
and is in NP $\cap$ coNP~\cite{BFLMS08,CAHS03,ZP96} 
(so the hardness result cannot be obtained for $k=1$). 


\begin{lemma}\label{lemm_memless_2}
The unknown intial credit and the mean-payoff threshold problems for 
multi-weighted two-player game structures for memoryless strategies for 
Player~1 are NP-hard, even in one-player game structures with dimension $k=2$ for
the weight function.
\end{lemma}

\begin{proof}   
We present a reduction from the {\sc Knapsack} problem. 
The {\sc Knapsack} problem consists of a set $I=\{1,2, \ldots,n\}$ of 
$n$ items, for each item $i$ there is a profit $p_i \in \mathN$ and a weight 
$w_i \in \mathN$.
Given a weight bound $B$ and profit bound $P$, the {\sc Knapsack} problem 
asks whether there exists a subset $J \subseteq I$ of items such that 
(a)~$\sum_{j \in J} w_j \leq B$; and (b)~$\sum_{j \in J} p_j \geq P$ 
(i.e., a profit of $P$ can be accumulated without exceeding weight $B$).
The {\sc Knapsack} problem is NP-hard~\cite{PapaBook}. 
Our reduction is as follows: given an instance of the {\sc Knapsack} problem 
we construct a one-player game structure with a weight function of dimension 
$2$.
The set of positions is as follows: 
$S_1 =I \cup \{(i,j) | i \in I, j \in\{Y,N\}\} \cup \{n+1\}$ and $S_2 = \emptyset$.
The set of edges is as follows: 
$E=\{(i,(i,Y)), (i,(i,N)) \mid i \in I\} \cup 
\{((i,Y),i+1), ((i,N),i+1) \mid i \in I\} \cup \{(n+1,1)\}.$
Intuitively, in the game structure, for every item Player~1 has a choice of 
``Yes" (edge from $i$ to $(i,Y)$) to select item $i$, and choice of ``No" 
(edge from $i$ to $(i,N)$) to not select item $i$. 
From $(i,Y)$ and $(i,N)$ the next position is $i+1$, and from the position 
$n+1$ the next position is~1.
The weight function function $w:E \to \integ^2$ has two dimensions: 
(a)~for edge $e=(i,(i,N))$ we have $w(e)=(0,0)$ (i.e., for the choice of ``No" 
all the weights are~0);
(b)~for an edge $e=(i,(i,Y))$ we have $w(e)=(p_i,-w_i)$ (i.e., for the choice 
of ``Yes", the first component gains the profit and the second component 
loses the weight of item $i$);
(c)~for an edge $e=((i,Y),i+1)$ or $e=((i,N),i+1)$ we have $w(e)=0$; and 
(d)~for the edge $e=(n+1,1)$ we have $w(e)=(-P,B)$ (i.e., there is a loss of 
$P$ in the first component and a gain of $B$ in the second component).
The construction is illustrated in Fig~\ref{figure:knapsack}.
Given a solution $J$ for the {\sc Knapsack} problem, the memoryless strategy 
that choose $(j,(j,Y))$ for $j \in J$, and $(j',(j',N))$ for 
$j' \in I \setminus J$, with intial credit $(0,B)$ is a solution for the 
unknown initial credit problem.
Conversely, given a memoryless strategy $\lambda_1$ for the unknown initial credit problem,
the set $J=\{j \in I \mid \lambda_1(j)=(j,Y)\}$ is a solution to 
the {\sc Knapsack} problem. 
The argument for the mean-payoff threshold problem is analogous.
The result follows.
\end{proof}

\begin{figure}[ht]
\begin{center}
\unitlength = 3mm
\begin{picture}(53,16)(0,0)
\gasset{Nw=4,Nh=4,Nmr=2}

\node(x2?)(2,8){$1$}
\node(x2)(7,14){$(1,Y)$}
\node(!x2)(7,2){$(1,N)$}
\node(y2?)(12,8){$2$}
\node(y2)(17,14){$(2,Y)$}
\node(!y2)(17,2){$(2,N)$}
\node(z2?)(22,8){$3$}
\node(n2?)(32,8){$n$}
\node(n2)(37,14){$(n,Y)$}
\node(!n2)(37,2){$(n,N)$}
\node(fin2)(42,8){$n+1$}
\drawedge(x2?,x2){$(p_1,-w_1)$}
\drawedge[ELside=r](x2?,!x2){$(0,0)$}
\drawedge[ELside=r](x2,y2?){$(0,0)$}
\drawedge(!x2,y2?){$(0,0)$}
\drawedge(y2?,y2){$(p_2,-w_2)$}
\drawedge[ELside=r](y2?,!y2){$(0,0)$}
\drawedge[ELside=r](y2,z2?){$(0,0)$}
\drawedge(!y2,z2?){$(0,0)$}
\drawedge(n2?,n2){$(p_n,-w_n)$}
\drawedge[ELside=r](n2?,!n2){$(0,0)$}
\drawedge[ELside=r](n2,fin2){$(0,0)$}
\drawedge(!n2,fin2){$(0,0)$}
\drawedge[dash={0.5}0](z2?,n2?){}
\gasset{Nw=5,Nh=5,Nmr=2.5,curvedepth=15}
\drawline(44,8)(45.5,8)
\node[Nframe=n](edgef)(49,8){$(-P,B)$ to~1}
\end{picture}
\end{center}
\caption{{\sc Knapsack} Reduction.}
\label{figure:knapsack}
\end{figure}
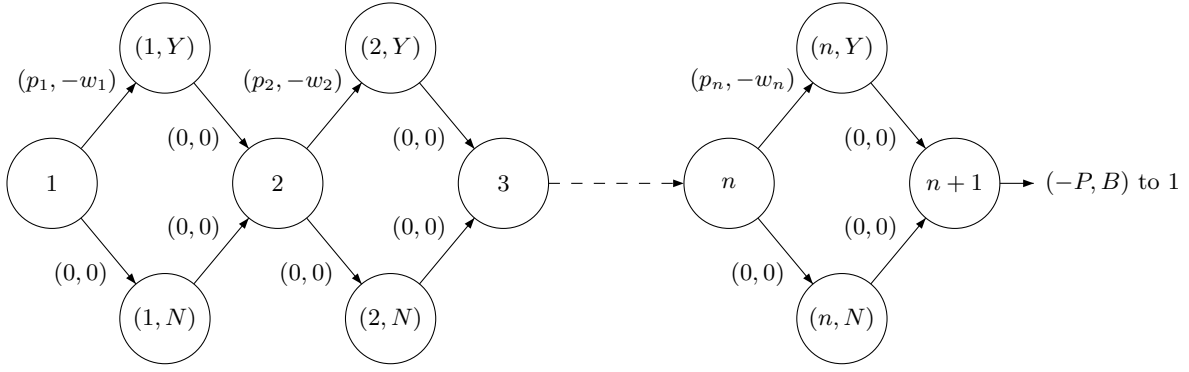

In Lemma~\ref{lemm_memless_3} we show the hardness of the 
problem when the weights are in $\{-1,0,1\}$, but the dimension is arbitrary.
It has been shown in~\cite{Cha10} that if the weights are $\{-1,0,1\}$ 
and the dimension is~2, then the problem can be solved in polynomial time.

\begin{lemma}\label{lemm_memless_3}
The unknown intial credit and the mean-payoff threshold problems for 
multi-weighted two-player game structures for memoryless strategies for 
Player~1 are NP-hard, even in one-player game structures when weights are 
restricted to $\{-1,0,1\}$.
\end{lemma}

\begin{proof} 
We present a reduction from the 3SAT problem. 
Consider a 3SAT formula $\Phi$ over a set $X=\{x_1,x_2,\ldots,x_n\}$ of 
variables, and a set $C_1,C_2,\ldots,C_m$ of clauses such that each 
clause has 3-literals (a literal is a variable or its complement).
We construct a one-player game structure with a weight function of 
dimension $m$ from $\Phi$. 
The set of positions is $S_1 = X \cup \{(x_i,j) \mid x_i \in X, j \in \{T,F\}\} \cup \{x_{n+1}\}$
and $S_2 = \emptyset$.
The set of edges is as follows: 
$E=\{(x_i,(x_i,T)), (x_i,(x_i,F)) \mid x_i \in X\} \cup 
\{((x_i,T),x_{i+1}), ((x_i,F),x_{i+1}) \mid x_i \in X\} \cup \{(x_{n+1},x_1)\}.$
Intuitively, in the game structure, for every variable Player~1 has a choice to set 
$x_i$ as ``True" (edge from $x_i$ to $(x_i,T)$), and choice to set $x_i$ as ``False" 
(edge from $x_i$ to $(x_i,F)$).
From $(x_i,T)$ and $(x_i,F)$ the next position is $x_{i+1}$, and from the 
position $x_{n+1}$ the next position is~$x_1$.
The construction of the graph is similar as in Fig~\ref{figure:knapsack}.
The weight function $w:E \to \integ^m$ has $m$ dimensions: 
(a)~for an edge $e=(x_i,(x_i,T))$ (resp. $e=(x_i,(x_i,F))$) and $1\leq k \leq m$, 
the $k$-th component of $w(e)$ is~1 if the choice $x_i$ as ``True"  (resp. ``False") 
satisfies clause $C_k$, and otherwise the $k$-th component is~0;
(b)~for edges $e=((x_i,j),x_{i+1})$, with $j \in \{T,F\}$, every component of $w(e)$ is~0; and
(c)~for the edge $e=(x_{n+1},x_1)$, for all $1 \leq k \leq m$, the $k$-th component of 
$w(e)=-1$.
If $\Phi$ is satisfiable, then consider a satisfying assignment $A$, and we construct a
memoryless strategy $\lambda_1$ as follows: for a position $x_i$, if $A(x_i)$ is ``True", then 
choose $(x_i,T)$, otherwise choose $(x_i,F)$. 
The memoryless strategy $\lambda_1$ with initial credit vector $\{0\}^m$ ensures that the
answer to the unknown initial credit problem for memoryless strategies is ``Yes".
Conversely, if there is a memoryless strategy $\lambda_1$ for the unknown initial credit 
problem, then the memoryless strategy must satisfy every clause. 
A satisfying assignment $A$ for $\Phi$ is as follows: $A(x_i)$ is ``True" if $\lambda_1(x_i)=(x_i,T)$,
and ``False", otherwise.
It follows that $\Phi$ is satisfiable iff the answer to the unknown initial credit problem 
for memoryless strategies is ``Yes".
The argument for the mean-payoff threshold problem is analogous.
The desired result follows.
\end{proof}

The following theorem follows from the results of Lemma~\ref{lemm_memless_1}, Lemma~\ref{lemm_memless_2}
and Lemma~\ref{lemm_memless_3}.

\begin{theorem}\label{thrm_memless}
The unknown initial credit and the mean-payoff threshold problems for 
multi-weighted two-player game structures for memoryless strategies for 
Player~1 are NP-complete. 
\end{theorem}

\section{Conclusion}
In this work we considered games with multiple mean-payoff and energy 
objectives, and established determinacy under finite-memory, inter-reducibility of 
these two classes of games for finite-memory strategies, and improved the 
complexity bounds from EXPSPACE to coNP-complete. 

Two interesting problems are open: (A)~for generalized mean-payoff games, 
the winning strategies with infinite memory are more powerful than finite-memory strategies, 
and the complexity of solving generalized mean-payoff games with infinite-memory 
strategies remains open.
(B)~it is not knwon how to compute the exact or 
approximate Pareto curve (trade-off curve) for multi-objective mean-payoff and
energy games.

\smallskip\noindent{\bf Acknowledgement.} 
We are grateful to Jean Cardinal for pointing the reference~\cite{KS88}.

\bibliographystyle{plain}
\bibliography{main}

\end{document}